\documentclass[twocolumn]{svjour3}          
\smartqed  
\usepackage{graphicx}
\begin{document}

\title{Convolution and Correlation Theorems for Wigner-Ville Distribution Associated with the Quaternion Offset Linear
Canonical  Transform}
\subtitle{Convolution and Correlation Theorems for WVD  Associated with the QOLCT}

\titlerunning{Convolution and Correlation Theorems}        

\author{M. YOUNUS BHAT        \and
        AAMIR H. DAR 
}


\institute{M. YOUNUS BHAT \at
              Department of  Mathematical Sciences,  Islamic University of Science and Technology Awantipora, Pulwama, Jammu and Kashmir 192122, India.
              \email{gyounusg@gmail.com}           
           \and
           AAMIR H. DAR \at
               Department of  Mathematical Sciences,  Islamic University of Science and Technology Awantipora, Pulwama, Jammu and Kashmir 192122, India.
}

\date{Received: date / Accepted: date}

\maketitle

\begin{abstract}
The quaternion offset linear canonical transform(QOLCT) has gained much popularity in recent years because of its applications in many areas, including color
image and signal processing. At the same time  the applications of Wigner-Ville distribution (WVD) in signal analysis and image processing can not be excluded. In this paper we investigate the Winger-Ville Distribution associated with quaternion  offset linear canonical transform (WVD-QOLCT). Firstly, we propose the  definition
of the WVD-QOLCT, and then several important properties of newly defined WVD-QOLCT, such as nonlinearity, bounded, reconstruction formula, orthogonality relation and Plancherel formula are derived. Secondly a novel canonical convolution operator and a related correlation operator for WVD-QOLCT are proposed. Moreover, based on the proposed operators, the corresponding generalized convolution, correlation  theorems are studied.We also show that the convolution and  correlation  theorems of the QWVD and WVD-QLCT can be looked as a special case of our achieved results.
\keywords{Quaternion algebra \and Offset linear canonical transform \and  Quaternion offset linear canonical transform \and Wigner-Ville distribution \and Convolution \and Correlation\and  Modulation}
\subclass{11R52 \and  42C40 \and  42C30\and 43A30.}
\end{abstract}

\section{Introduction}
\label{intro}
In the time-frequency signal analysis the classical Wigner-Ville distribution (WVD) or Wigner- Ville transform (WVT) has an important role to play.  Eugene Wigner introduced the concept WVD while making his calculation of the quantum corrections. later on it was J. Ville who derived it independently as a quadratic representation of the local time-frequency energy of a signal in 1948. Many important properties of WVT has been studied by  many authors. On replacing the kernel of the classical Fourier transform (FT) with the kernel of the LCT in the WVD domain, this transform can be extended to the domain of linear canonical transform \cite{3}-\cite{6},  \cite{13}-\cite{19}. 

On the other hand the quaternion Fourier transform (QFT) is of the interest in the present era. Many important properties like shift, modulation, convolution, correlation, differentiation, energy conservation, uncertainty principle of QFT have been found. Many generalized transforms are closely related to the QFTs, for example, the quaternion wavelet transform, fractional quaternion Fourier transform, quaternion linear canonical transform, and quaternionic windowed Fourier transform. Based on the QFTs, one also may extend the WVD to the quaternion algebra while enjoying similar properties as in the classical case. Many authors  generalized  the classical WVD to quaternion algebra, which they called as  the quaternion Wigner-Ville distribution (QWVD). For more details we refer to \cite{1}, \cite{2}, \cite{7}-\cite{12}. 

The linear canonical transform (LCT) with four parameters $(a, b, c, d)$ has been generalized to a six parameter transform  $(a,b,c,d,u_0,w_0)$ known as offset linear canonical transform (OLCT). Due to the time shifting $u_0$ and frequency modulation parameters, the OLCT has gained more flexibility over classical LCT. Hence has found wide applications in image and signal processing.  On the other side the convolution has some applications in various areas of Mathematics like linear algebra, numerical analysis and signal processing. Where as Correlation like convolution is an another important tool n signal processing, optics and detection applications. In the domains of LCT, WVD and OLCT the convolution and correlation operations have been studied \cite{7}-\cite{10}.

The quaternion offset linear canonical transform\\ (QOLCT) has gained much popularity in recent years because of its applications in many areas, including colour
image and signal processing. At the same time  the applications of Wigner-Ville distribution (WVD) in signal analysis and image processing can not be excluded. Motivated by QOLCT and WVD, we in this paper we investigate the Winger-Ville Distribution associated with quaternion  offset linear canonical transform (WVD-QOLCT). Firstly, we propose the  definition of the WVD-QOLCT, and then several important properties of newly defined WVD-QOLCT, such as nonlinearity, bounded, reconstruction formula, orthogonality relation and Plancherel formula are derived. Secondly a novel canonical convolution operator and a related correlation operator for WVD-QOLCT are proposed. Moreover, based on the proposed operators, the corresponding generalized convolution, correlation  theorems are studied.We also show that the convolution and  correlation  theorems of the QWVD and WVD-QLCT can be looked as a special case of our achieved results.

The paper is organised as follows. In Section \ref{sec 3}, we provide  the definition of  Wigner-Ville distribution associated with the quaternionic offset linear canonical transform (WVD-QOLCT). Then we will investigate several basic properties of the WVD-QOLCT  which are important for signal representation in signal processing. In Section \ref{sec 4} we first define the convolution and correlation for the QOLCT. We then establish the new convolution and correlation for the WVD-QOLCT.We also show that the convolution theorems of the QWVD and WVD-QLCT can be looked as a special case of our achieved results
I
\section{Winger-ville Distribution associated with  Quaternion Offset Linear \\ Canonical Transform(WVD-QOLCT)}\label{sec 3}
 Since in practice most natural signals are non-stationary. In order to study a non-stationary signals the
Wigner-Ville distribution has become a suite tool for the analysis of the non stationary
signals. In this section,we are going to give the definition of  Wigner-Ville distribution associated
with the quaternionic offset linear canonical transform (WVD-QOLCT),then we will investigate several basic properties
of the WVD-QOLCT  which are important for signal representation
in signal processing.
\begin{definition}\label{def 3.1}Let $A_i =\left[\begin{array}{cccc}a_i & b_i &| & r_i\\c_s & d_i &| & s_i \\\end{array}\right] $, be a matrix parameter such that $a_s$, $b_i$, $c_i$, $d_i$, $r_i$, $s_i \in \mathbf R$ and $ a_id_i-b_ic_i=1,$ for $i=1,2.$ The Wigner-Ville distribution associated with the two-sided quaternionic offset linear
canonical transform (WVD-QOLCT) of  signals $f,g\in L^2(\mathbf R^2,\mathbf H)$, is given by\\\\
\begin{equation}\label{3.1}
{\mathcal W}_{f,g }^{A_1,A_2}(t,u)=
\left\{
  \begin{array}{ll}
   \int_{{\mathbf R}^2}{K^{i }_{A_1}(n_1,u_1)}f(t+\frac{n}{2})\overline{g(t-\frac{n}{2})}\\\
   \quad K^{j }_{A_2}(n_2,u_2)dn,\\\
   \quad b_1,b_2\ne 0,   \\\
    \sqrt{d_1}{e}^{i (\frac{c_1d_1}{2}{\left(u_1-r_1\right)}^2+{\ u}_1r_1)}\\\
    \quad f(t_1+\frac{d_1{(u}_1-r_1)}{2},t_2+\frac{n_2}{2})\\\
\quad \overline{g}(t_1-\frac{d_1{(u}_1-r_1)}{2},t_2-\frac{n_2}{2})\\\
K^{j }_{A_2}\left(n_2,u_2\right),\\\
 \quad b_1=0,b_2\ne 0; \\\
    \sqrt{d_2}{K^{i }_{A_1}\left(n_1,u_1\right)}\\\
    \quad f(t_1+\frac{n_1}{2},t_2+\frac{d_2{(u}_2-r_2)}{2})\\\
     \quad \overline{g}(t_1-\frac{n_1}{2},t_2-\frac{d_2{(u}_2-r_2)}{2})\\\
     \quad   e^{j (\frac{c_2d_2}{2}{(u_2-r_2)}^2+{\ u}_2{r }_{2)}},\\\
       \quad b_1\ne 0, b_2=0; \\\
    \sqrt{d_1d_2}e^{i(\frac{c_1d_1}{2}{\left(u_1-r_1\right)}^2+{u}_1r_1)}\\\
   \quad  f(t_1+\frac{d_1({u}_1-{r}_1}{2}),t_2+\frac{d_2(u_2-r_2}{2})) \\\
    \quad  \overline{g} (t_1-\frac{d_1({u}_1-{r}_1}{2}),t_2-\frac{d_2(u_2-r_2}{2}))\\\
      \quad   e^{j(\frac{c_2d_2}2{(u_2-r_2)}^2+{u}_2{r}_2)}, \\\
         \quad b_1=b_2=0.
  \end{array}
\right.\end{equation}
where $t=(t_{1},t_{2}), u= (u_1,u_2),n=(n_1,n_2)$ and  $K_{A_1}^i(n_1,u_1)$ and $K_{A_2}^j(n_2,u_2)$ are the quaternion kernels.
\end{definition}
\begin{note}
 If $f=g$ then $\mathcal{W}_{f,f }^{A_1,A_2}(t,u)$ we call it the Auto WVD-QOLCT.Otherwise is is called Cross  WVD-QOLCT\end{note}

Without loss of generality we will deal with the case $b_i\neq 0,\, i=1,2,$ as in other cases proposed transform reduces to a chrip multiplications.Thus for any $f,g\in L^2(\mathbf R^2,\mathbf H)$ we have
\begin{eqnarray}
\nonumber\mathcal{W}_{f,g }^{A_1,A_2}(t,u)&=&\int_{{\mathbf R}^2}{K^{i }_{A_1}(n_1,u_1)}f(t+\frac{n}{2})\overline{g(t-\frac{n}{2})}K^{j }_{A_2}(n_2,u_2)dn\\\
\nonumber&=&\mathcal{O}^{i,j}_{A_1,A_2}\left\{f\left( t+\frac{n}{2}\right)\overline{g\left(t-\frac{n}{2}\right)}\right\}\\\
\label{3.2}&=&\mathcal{O}^{i,j}_{A_1,A_2}\{h_{f,g}(t,n)\}.
\end{eqnarray}
Where $h_{f,g}(t,n)=f\left( t+\frac{n}{2}\right)\overline{g\left(t-\frac{n}{2}\right)}$ is known as quaternion correlation product.
Applying the inverse QOLCT to (\ref{3.2}), we get
$$\{h_{f,g}(t,n)\}=\{\mathcal O_{A_1,A_2}^{i,j}\}^{-1}\{\mathcal  W_{f,g}^{A_1,A_2}(t,u)\}$$
which implies
\begin{eqnarray}
\nonumber f\left(t+\frac{n}{2}\right)\overline{g\left(t-\frac{n}{2}\right)}&=&\{\mathcal O_{A_1,A_2}^{i,j}\}^{-1}\{\mathcal  W_{f,g}^{A_1,A_2}(t,u)\}\\\
\nonumber&=&\int_{\mathbf R^2} K_{A_1}^{-i}(t_1,u_1) \mathcal W_{f,g}^{A_1,A_2}(t,u)\\\
\label{3.3}&&K_{A_2}^{-j}(t_2,u_2)dw.
\end{eqnarray}
Now, we discuss several basic properties of the WVD-QOLCT given by (\ref{3.1}). These properties play important roles in signal representation.\\
\begin{theorem}\label{th 3.2}
{\bf (Boundedness)} Let $f,g\in L^2(\mathbf{R}^2,\mathbf{H}).$ Then
 \begin{equation}\label{3.4}\big|\mathcal W_{f,g}^{A_1,A_2}(t,u)\big|\le\frac{2}{\pi\sqrt{b_1b_2}}\|f\|_{L^{2}\mathbf({R}^2,H)}\|g\|_{L^{2}\mathbf({R}^2,H)}\end{equation}
\end{theorem}
\begin{proof} By the virtue of  Cauchy-Schwarz inequality in quaternion domain, we have
\begin{eqnarray*}
&&|\mathcal W_{f,g}^{A_1,A_2}(t,u)|^{2}\\\
&&=\left|\int_{\mathbf R^2} K_{A_1}^i(n_1,u_1) \,f\left(t+\frac{n}{2}\right)\overline{g\left(t-\frac{n}{2}\right)}\,K_{A_2}^j(n_2,u_2)dn\right| ^2\\\
&&\leq \left( \int_{\mathbf{R}^2} \left| K_{A_1}^{\mathbf{i}}(n_1,u_1)f\left(t+\frac{n}{2}\right)\overline{g\left(t-\frac{n}{2}\right)}K_{A_2}^{\mathbf{j}}(n_2,u_2)\right| \rm{d}n\right) ^2 \\\
&&=\left( \frac{1}{\sqrt{4\pi^{2}|b_{1}b_{2}|}}\int_{\mathbf{R}^2} \left| f\left(t+\frac{n}{2}\right)\overline{g\left(t-\frac{n}{2}\right)}\right| \rm{d}n\right) ^2 \\\
&&\leq \frac{1}{4\pi^{2}|b_{1}b_{2}|}\left(\int_{\mathbf{R}^2} \left| f\left(t+\frac{n}{2}\right)\right|^{2}\rm{d}n\right)\left(\int_{\mathbf{R}^2} \left| \overline{g\left(t-\frac{n}{2}\right)}\right|^{2} \rm{d}s\right)\\\
&&=\frac{1}{4\pi^{2}|b_{1}b_{2}|}\left(4\int_{\mathbf{R}^2} \left| f(w)\right|^{2}\rm{d}w\right)\left(4\int_{\mathbf{R}^2} \left| \overline{g(y)}\right|^{2} \rm{d}y\right)\\\
&&=\frac{4}{\pi^{2}|b_{1}b_{2}|} \|f\|_{L^{2}\mathbf({R}^2,H)}^{2} \|g\|_{L^{2}\mathbf({R}^2,H)}^{2}
\end{eqnarray*}
where applying the change of variables $w=t+\frac{n}{2}$ and  $y=t-\frac{n}{2}$ in the last second  step. Then we have
\begin{eqnarray*}
|\mathcal W_{f,g}^{A_1,A_2}(t,u)|\leq \frac{2}{\pi\sqrt{|b_{1}b_{2}|}} \|f\|_{L^{2}\mathbf({R}^2,H)} \|g\|_{L^{2}\mathbf({R}^2,H)}
\end{eqnarray*}
which completes the proof of Theorem. \end{proof}
\begin{theorem}\label{th 3.3}{\bf (Nonlinearity)}  Let $f$ and $g$ be two quaternion functions in $L^2(\mathbf R^2,\mathbf H).$ Then
\begin{equation}\label{3.5}\mathcal W_{f+g}^{A_1,A_2}=\mathcal W_{f,f}^{A_1,A_2}+\mathcal W_{f,g}^{A_1,A_2}+\mathcal W_{g,f}^{A_1,A_2}+\mathcal W_{g,g}^{A_1,A_2}\end{equation}
\end{theorem}
\begin{proof} By definition \ref{def 3.1} we have
\begin{eqnarray*}
&&\mathcal W_{f+g}^{A_1,A_2}(t,u)\\\
&&=\int_{{\mathbf R}^2}{K^{i }_{A_1}(n_1,u_1)}\left[f(t+\frac{n}{2})+{g}(t+\frac{n}{2})\right]\\\
&&\overline{\left[f(t-\frac{n}{2})+{g}(t-\frac{n}{2})\right]}K^{j }_{A_2}(n_2,u_2)dn\\\\
&&=\int_{{\mathbf R}^2}{K^{i }_{A_1}(n_1,u_1)}\left[f(t+\frac{n}{2})\overline{f(t-\frac{n}{2})}\right.\\\
&&+f(t+\frac{n}{2})\overline{g(t-\frac{n}{2})}\\\
&& +\left.g(t+\frac{n}{2})\overline{f(t-\frac{n}{2})}+g(t+\frac{n}{2})\overline{g(t-\frac{n}{2})}\right]
K^{j }_{A_2}(n_2,u_2)dn\\\\
&&=\int_{{\mathbf R}^2}{K^{i }_{A_1}(n_1,u_1)}f\left(t+\frac{n}{2}\right)\overline{f\left(t-\frac{n}{2}\right)}K^{j }_{A_2}(n_2,u_2)dn\\\
&&\;+\int_{{\mathbf R}^2}{K^{i }_{A_1}(n_1,u_1)}f\left(t+\frac{n}{2}\right)\overline{g\left(t-\frac{n}{2}\right)}K^{j }_{A_2}(n_2,u_2)dn\\\
&&\;+\int_{{\mathbf R}^2}{K^{i }_{A_1}(n_1,u_1)}g\left(t+\frac{n}{2}\right)\overline{f\left(t-\frac{n}{2}\right)}K^{j }_{A_2}(n_2,u_2)dn\\\
&&\;+\int_{{\mathbf R}^2}{K^{i }_{A_1}(n_1,u_1)}g\left(t+\frac{n}{2}\right)\overline{g\left(t-\frac{n}{2}\right)}K^{j }_{A_2}(n_2,u_2)dn\\\\
&&=\mathcal W_{f,f}^{A_1,A_2}+\mathcal W_{f,g}^{A_1,A_2}+\mathcal W_{g,f}^{A_1,A_2}+\mathcal W_{g,g}^{A_1,A_2}
\end{eqnarray*}
which completes the proof of Theorem. \end{proof}
Note the properties like Shift,Modulation,Dilation are similar to the classical QOLCT so we avoided them.
\begin{theorem}\label{th 3.4}
{\bf (Reconstruction formula).} For  $f,g\in L^{2}(\mathbf{R}^2,\mathbf{H})$ where $g$ does not vanish at $0$ .We
get the following inversion formula of the WVD-QOLCT:
\begin{eqnarray}\nonumber f(t)&=&\frac{1}{\overline{g(0)}}\int_{\mathbf R^2}K^{-i }_{A_1}(u_1,n_1)\mathcal W_{f,g}^{A_1,A_2}\left(\frac{t}{2},u\right)\\\
\label{3.6}&&K^{-j }_{A_2}(u_2,n_2)du\end{eqnarray}
\end{theorem}
\begin{proof} By (\ref{3.3}), we have
 \begin{eqnarray*}
\{h_{f,g}(t,n)\}&=\{\mathcal O_{A_1,A_2}^{i,j}\}^{-1}\{\mathcal  W_{f,g}^{A_1,A_2}(t,u) \}\\
\end{eqnarray*}
which implies
\begin{eqnarray*}
&&f\left(t+\frac{n}{2}\right)\overline{g\left(t-\frac{n}{2}\right)}\\\
&&=\int_{\mathbf R^2} K_{A_1}^{-i}(t_1,u_1) \mathcal W_{f,g}^{A_1,A_2}(t,u)K_{A_2}^{-j}(t_2,u_2)dw,
\end{eqnarray*}
Now let $t=\frac{n}{2}$ and taking change of variable $w=2t$, we get
$$f(w)=\frac{1}{\overline{g(0)}}\int_{\mathbf R^2}K^{-i }_{A_1}(u_1,n_1)\mathcal W_{f,g}^{A_1,A_2}\left(\frac{w}{2},u\right)K^{-j }_{A_2}(u_2,n_2)du$$
which completes the proof of Theorem. \end{proof}
\begin{theorem}\label{th 3.5}
{\bf (Orthogonality relation).}  If $f_1,f_2,g_1,g_2 \in L^2(\mathbf{R}^2,\mathbf{H})$ are quaternion-valued signals.Then
\begin{equation}\label{3.7}\left\langle \mathcal W^{A_1,A_2}_{f_1,g_1}(t,u), {\mathcal W^{A_1,A_2}_{f_2,g_2}(t,u)} \right\rangle
=\left[\left\langle f_1,f_2\right\rangle\left\langle g_2,g_1\right\rangle\right]_{\mathbf H}\end{equation}
\end{theorem}		
\begin{proof} By the definition of Winger-ville distribution associated with quaternion OLCT and innear product relation we have
\begin{eqnarray}
\nonumber &&\langle W^{A_1,A_2}_{f_1,g_1}(t,u), {\mathcal W^{A_1,A_2}_{f_2,g_2}(t,u)} \rangle\\\
\nonumber&&=\int_{\mathbf{R}^4}\left[ \mathcal W^{A_1,A_2}_{f_1,g_1}(t,u) \overline{{\mathcal W^{A_1,A_2}_{f_2,g_2}(t,u)}}\right]_{\mathbf H}dudt\\\
\nonumber&&=\int_{\mathbf{R}^4}\left[ \mathcal W^{A_1,A_2}_{f_1,g_1}(t,u) \right.\\\
\nonumber&&\left.\overline{\int_{\mathbf{R}^2}K_{A_1}^i(n_1,u_1)f_2\left(t+\frac{n}{2}\right) \overline{g_2\left(t-\frac{n}{2}\right)}} \right.\\\
\nonumber&&\left.\overline{ K_{A_2}^j(n_2,u_2)dn}\right]_{\mathbf H}dudt\\\
\nonumber&&=\int_{\mathbf{R}^6}\left[ \mathcal W^{A_1,A_2}_{f_1,g_1}(t,u) K_{A_2}^{-j}(n_2,u_2)g_2\left(t-\frac{n}{2}\right)\right.\\\
\nonumber&&\left. \overline{f_2\left(t+\frac{n}{2}\right)}K_{A_1}^{-i}(n_1,u_1) \right]_{\mathbf H}dudtdn\\\
\nonumber&&=\int_{\mathbf{R}^6}\left[ K_{A_1}^{-i}(n_1,u_1) \mathcal W^{A_1,A_2}_{f_1g_1}(t,u) K_{A_2}^{-j}(n_2,u_2)\right.\\\
\nonumber&&\left. g_2\left(t-\frac{n}{2}\right)\overline{f_2\left(t+\frac{n}{2}\right)} \right]_{\mathbf H}dudtdn\\\
\nonumber&&=\int_{\mathbf{R}^4}\left[ \int_{\mathbf{R}^2}K_{A_1}^{-i}(n_1,u_1) \mathcal W^{A_1,A_2}_{f_1,g_1}(t,u) K_{A_2}^{-j}(n_2,u_2)du\right.\\\
\label{3.8}&&\left. g_2\left(t-\frac{n}{2}\right) \overline{f_2\left(t+\frac{n}{2}\right)} \right]_{\mathbf H}dtdn
\end{eqnarray}
Because
\begin{eqnarray*} \overline{K_{A_1}^{i}(n_1,u_1)}=K_{A_{1}}^{-i}(u_1,n_1)=K_{A_{1}^{-1}}^{i}(u_1,n_1)\end{eqnarray*}
\begin{eqnarray*} \overline{K_{A_2}^{j}(n_2,u_2)}=K_{A_{2}}^{-j}(n_2,u_2)=K_{A_{2}^{-1}}^{j}(u_2,n_2)\end{eqnarray*}	
Now by using (\ref{3.3}) in (\ref{3.8}), we have
\begin{eqnarray*}
&& \left\langle\mathcal W^{A_1,A_2}_{f_1,g_1}(t,u), {\mathcal W^{A_1,A_2}_{f_2,g_2}(t,u)} \right\rangle\\\
&&=\int_{\mathbf{R}^4}\left[\int_{\mathbf{R}^2}K_{A_1^{-1}}^{i}(u_1,n_1)\mathcal W^{A_1,A_2}_{f_1,g_1}(t,u) K_{A_2^{-1}}^{j}(u_2,n_2)du\right.\\\
&&\left. g_2\left(t-\frac{n}{2}\right)\overline{f_2\left(t+\frac{n}{2}\right)}\right]_{\mathbf H}dtdn\\\
&&=\int_{\mathbf{R}^4}\left[ f_1\left(t+\frac{n}{2}\right)\overline{g_1\left(t-\frac{n}{2}\right)} g_2\left(t-\frac{n}{2}\right)\overline{f_2\left(t+\frac{n}{2}\right)}  \right]_{\mathbf H}dtdn
\end{eqnarray*}
Using the change of variables $t+\frac{n}{2}=\omega$,and $t-\frac{n}{2}=\xi$  the equation becomes
\begin{eqnarray*}
&&\langle\mathcal W^{A_1,A_2}_{f_1,g_1}(t,u), \mathcal W^{A_1,A_2}_{f_2,g_2}(t,u) \rangle\\\
&&=\int_{\mathbf{R}^4}\left[ f_1(\omega)\overline{g_1(\xi)} g_2(\xi)\overline{f_2(\omega)}\right]_{\mathbf H}d\omega d\xi\\\
&&=\left[\int_{\mathbf{R}^2}f_1(\omega) \overline{f_2(\omega)} d\omega \int_{\mathbf{R}^2}g_2(\xi) \overline{g_1(\xi)} d\xi \right]_{\mathbf H}\\\
&&=[\langle f_1,f_2\rangle\langle g_2,g_1\rangle]_{\mathbf H}
\end{eqnarray*}
which completes the proof theorem. \end{proof}
{\bf Consequences of Theorem \ref{th 3.5}.}
\begin{enumerate}
\item  If $g_1=g_2=g$, then
\begin{equation}\label{3.9}\langle \mathcal W^{A_1,A_2}_{f_1,g}(t,u), {\mathcal W^{A_1,A_2}_{f_2,g}(w,u)} \rangle
=\|g\|^{2}_{L^2(\mathbf{R}^2)}\langle f_1,f_2\rangle\end{equation}
\item  If $f_1=f_2=f$, then
\begin{equation}\label{3.10}\langle \mathcal W^{A_1,A_2}_{f,g_1}(t,u), {\mathcal W^{A_1,A_2}_{f,g_2}(w,u)} \rangle
=\|f\|^{2}_{L^2(\mathbf{R}^2)}\langle g_1,g_2\rangle.\end{equation}
\item If  $f_1=f_2=f$ and $g_1=g_2=g$, then
\begin{eqnarray}\nonumber &&\langle \mathcal W^{A_1,A_2}_{f,g}(t,u), {\mathcal W^{A_1,A_2}_{f,g}(w,u)} \rangle\\\
\nonumber && =\int_{\mathbf{R}^2}\int_{\mathbf{R}^2}| \mathcal W^{A_1,A_2}_{f,g}(t,u)|^{2}dudt\\\
\label{3.11}&&=\|f\|^{2}_{L^2(\mathbf{R}^2)} \|g\|^{2}_{L^2(\mathbf{R}^2)}\end{eqnarray}
\end{enumerate}
\begin{theorem}\label{th 3.6}
{\bf (Plancherel’s theorem for\\ WVD-QOLCT)}.For $f, g\in L^{2}(\mathbf{R}^2,\mathbf{H})$, we have the equality
\begin{eqnarray}
\nonumber &&\int_{\mathbf{R}^2}\int_{\mathbf{R}^2}|\mathcal W^{A_1,A_2}_{f,g}(t,u)|^2 dudt\\\
\nonumber &&=\|\mathcal W^{A_1,A_2}_{f,g}\|^2 _{L^{2}(\mathbf{R}^2,\mathbf{H})}\\\
\label{3.12}&&=\|f\|^2_{L^{2}(\mathbf{R}^2,\mathbf{H})}\|g|^2_{L^{2}(\mathbf{R}^2,\mathbf{H})}\end{eqnarray}
\end{theorem}		
\begin{proof} If we look at  (\ref{3.11}), the proof of the theorem follows.\end{proof}

Now we move forward towards our main section that is convolution and correlation theorems for winger-ville distribution associated with quaternion offset linear canonical transform.
\section{ Convolution and Correlation theorem for WVD-QOLCT}\label{sec 4}
The convolution and correlation are fundamental signal processing algorithms in the theory of linear time-invariant(LTI) systems. In engineering, they have been widely used for various template matchings. In the following we first define the convolution and correlation for the QOLCT. They are extensions of the convolution definition from the OLCT (see \cite{16}) to the QOLCT domain. We then establish the new convolution and correlation for the WVD-QOLCT.We also show that the convolution theorems of the QWVD and WVD-QLCT can be looked as a special case of our achieved results.
\begin{definition}\label{def 4.1} For any two quaternion functions $f,g \in L^2(\mathbf R^2,\mathbf H),$ we define the convolution operator of the QOLCT
as
\begin{equation}\label{4.1}(f\star g)(t)=\int_{\mathbf R^2}\Psi(z_1,t_1)f(z)g(t-z)\Psi(z_2,t_2) dz\end{equation}
Where $\Psi(z_1,t_1)$ and $\Psi(z_2,t_2)$ are known as weight functions.\\
 We assume
\begin{eqnarray*}\Psi(z_1,t_1)=e^{-i\frac{a_1}{b_1}2z_1(t_1-z_1)}\end{eqnarray*}
and
\begin{equation}\label{4.2}\Psi(z_2,t_2)=e^{-j\frac{a_2}{b_2}2z_2(t_2-z_2)}\end{equation}
\end{definition}
As a consequence of the above definition, we get the following important theorem.
\begin{theorem}\label{th 4.2}{(WVD-QOLCT Convolution).} For any two quaternion functions $f,g \in L^2{(\mathbf R^2,\mathbf H)},$ the following result holds
\begin{eqnarray}
\nonumber &&\mathcal W^{A_1,A_2}_{f\star g}(t,u)\\\
\nonumber&&=\sqrt{2\pi b_1 i}e^{\frac{-i}{2b_1}[d_1(u_1^2+r_1^2)-2u_1(d_1r_1-b_1s_1)]}\\\
\nonumber&&\times\left\{\int_{\mathbf R^2}e^{-i\frac{a_1}{b_1}(4w_1(t_1-w_1))}\mathcal W^{A_1,A_2}_{f,f}(w,u)\mathcal W^{A_1,A_2}_{g,g}(t-w,u)\right.\\\
\nonumber&&\left.e^{-j\frac{a_2}{b_2}(4w_2(t_2-w_2))}dw\right\}\\\
\label{4.3}&&\times\sqrt{2\pi b_2 j}e^{\frac{-j}{2b_2}[d_2(u_2^2+r_2^2)-2u_2(d_2r_2-b_2s_2)]}
\end{eqnarray}
\end{theorem}
\begin{proof} Applying the definition of the WVD-QOLCT we have
\begin{eqnarray}
\nonumber \mathcal W^{A_1,A_2}_{f\star g}(t,u)&=&\int_{{\mathbf R}^2}{K^{i }_{A_1}(n_1,u_1)}\left[(f\star g)(t+\frac{n}{2})\right]\\\
\label{4.4}&&\left[\overline{f}\star\overline{g}(t-\frac{n}{2})\right]K^{j }_{A_2}(n_2,u_2)dn\end{eqnarray}
Now using Definition \ref{def 4.1} in (\ref{4.4}) we have
\begin{eqnarray}
\nonumber &&\mathcal W^{A_1,A_2}_{f\star g}(t,u)\\\
\nonumber &&=\int_{{\mathbf R}^2}{K^{i }_{A_1}(n_1,u_1)}\left\{\int_{{\mathbf R}^2}\Psi_1(z_1,t_1+\frac{n_1}{2})f(z)\right.\\\
\nonumber &&\qquad g(t+\frac{n}{2}-z) \Psi_2(z_2,t_2+\frac{n_2}{2})dz\\\
\nonumber &&\qquad\left.\times\int_{{\mathbf R}^2}\Psi_1(\gamma_1,t_1-\frac{n_1}{2})\overline {f(\gamma)g(t-\frac{n}{2}-\gamma)}\right. \\\
\nonumber &&\qquad \left. \Psi_2(\gamma_2,t_2-\frac{n_2}{2})d\gamma \right\}{K^{j }_{A_2}(n_2,u_2)}dn\\\
\nonumber &&=\int_{{\mathbf R}^2}{K^{i }_{A_1}(n_1,u_1)}\left\{\int_{{\mathbf R}^2}e^{-i\frac{a_1}{b_1}2z_1((t_1+\frac{n_1}{2})-z_1)}f(z)\right.\\\
\nonumber &&\qquad  g(t+\frac{n}{2}-z) e^{-j\frac{a_2}{b_2}2z_2((t_2+\frac{n_2}{2})-z_2)}dz\\\
\nonumber&&\qquad \times\int_{{\mathbf R}^2}e^{-i\frac{a_1}{b_1}2\gamma_1((t_1-\frac{n_1}{2})-\gamma_1)}\overline{f(\gamma)g(t-\frac{n}{2}-\gamma)}\\\
\label{4.5} &&\qquad \left. e^{-j\frac{a_2}{b_2}2\gamma_2((t_2-\frac{n_2}{2})-\gamma_2)}d\gamma \right\}{K^{j }_{A_2}(n_2,u_2)}dn
\end{eqnarray}
 For simplicity let us denote
\begin{eqnarray}
\nonumber K^{i }_{A_1}(t_1,u_1)&=&K^i_{A_1}e^{\frac{i}{2b_1}[a_1t_1^2+2t_1(r_1-u_1)-2u_1(d_1r_1-b_1s_1)+d_1u_1^2]},\\\
\label{4.6}&&K^i_{A_1}=\frac{1}{\sqrt{2\pi b_1 i}}e^{i\frac{d_1}{2b_1}r_1^2}\end{eqnarray}
and
\begin{eqnarray}
\nonumber K^{j}_{A_2}(t_2,u_2)&=&K^j_{A_2}e^{\frac{j}{2b_2}[a_2t_2^2+2t_2(r_2-u_2)-2u_2(d_2r_2-b_2s_2)+d_2u_2^2]},\\\
\label{4.7}&& K^j_{A_2}=\frac{1}{\sqrt{2\pi b_2 j}}e^{j\frac{d_2}{2b_2}r_2^2}\end{eqnarray}
Now using (\ref{4.6})and (\ref{4.7}) in (\ref{4.5}),we have
\begin{eqnarray*}
&&\mathcal W^{A_1,A_2}_{f\star g}(t,u)\\\
&&=\int_{{\mathbf R^6}}K^i_{A_1}e^{\frac{i}{2b_1}[a_1n_1^2+2n_1(r_1-u_1)-2u_1(d_1r_1-b_1s_1)+d_1u_1^2]}\\\
&& \;e^{-i\frac{a_1}{b_1}2z_1((t_1+\frac{n_1}{2})-z_1)}\\\
&&\;\times f(z)g(t+\frac{n}{2}-z)
 e^{-j\frac{a_2}{b_2}2z_2((t_2+\frac{n_2}{2})-z_2)}\\\
 && \;e^{-i\frac{a_1}{b_1}2\gamma_1((t_1-\frac{n_1}{2})-\gamma_1)}\\\
&& \; \times \overline{f(\gamma)g(t-\frac{n}{2}-\gamma)}e^{-j\frac{a_2}{b_2}2\gamma_2((t_2-\frac{n_2}{2})-\gamma_2)}\\
&& \; \times K^j_{A_2}e^{\frac{j}{2b_2}[a_2n_2^2+2n_2(r_2-u_2)-2u_2(d_2r_2-b_2s_2)+d_2u_2^2]}dzd\gamma dn
\end{eqnarray*}
Setting $z_i=w_i+\frac{p_i}{2},\gamma_i=w_i-\frac{p_i}{2},i=1,2$ we get
\begin{eqnarray*}
&&\mathcal W^{A_1,A_2}_{f\star g}(t,u)\\\
&&=\int_{{\mathbf R^6}}K^i_{A_1}e^{\frac{i}{2b_1}[a_1n_1^2+2n_1(r_1-u_1)-2u_1(d_1r_1-b_1s_1)+d_1u_1^2]}\\\
&&\qquad e^{-i\frac{a_1}{b_1}2\left(w_1+\frac{p_1}{2}\right)\left((t_1+\frac{n_1}{2})-(w_1+\frac{p_1}{2})\right)}\\\
&&\qquad\times f\left(w+\frac{p}{2}\right)g\left(t+\frac{n}{2}-(w+\frac{p}{2})\right)\\\
&&\qquad e^{-j\frac{a_2}{b_2}2\left(w_2+\frac{p_2}{2}\right)\left((t_2+\frac{n_2}{2})-(w_2+\frac{p_2}{2})\right)}\\\
&&\qquad\times e^{-i\frac{a_1}{b_1}2\left(w_1-\frac{p_1}{2}\right)\left((t_1-\frac{n_1}{2})-(w_1-\frac{p_1}{2})\right)}\\\
&& \qquad \overline{f\left(w-\frac{p}{2}\right)g\left(t-\frac{n}{2}-(w-\frac{p}{2})\right)}\\\
&&\qquad\times e^{-j\frac{a_2}{b_2}2\left(w_2-\frac{p_2}{2}\right)\left((t_2-\frac{n_2}{2})-(w_2-\frac{p_2}{2})\right)}\\\
&&\; K^j_{A_2}e^{\frac{j}{2b_2}[a_2n_2^2+2n_2(r_2-u_2)-2u_2(d_2r_2-b_2s_2)+d_2u_2^2]}dpdq dw
\end{eqnarray*}
and $n_i=p_i+q_i,i=1,2$ we obtain
\begin{eqnarray}
\nonumber&&\mathcal W^{A_1,A_2}_{f\star g}(t,u)\\\
\nonumber&&=\int_{{\mathbf R^6}}\\\
\nonumber &&\; K^i_{A_1}e^{\frac{i}{2b_1}[a_1(p_1+q_1)^2+2(p_1+q_1)(r_1-u_1)-2u_1(d_1r_1-b_1s_1)+d_1u_1^2]}\\\
\nonumber &&\; e^{-i\frac{a_1}{b_1}\left(4w_1(t_1-w_1)\right)}e^{-i\frac{a_1}{b_1}p_1q_1}\\\
\nonumber &&\;\times f\left(w+\frac{p}{2}\right)\overline{f\left(w-\frac{p}{2}\right)}g\left(t-w+\frac{q}{2}\right)\overline{g\left(t-w-\frac{q}{2}\right)}\\\
\nonumber &&\; e^{-j\frac{a_2}{b_2}\left(4w_2(t_2-w_2)\right)}e^{-j\frac{a_2}{b_2}p_2q_2}\\\
\nonumber && \; \times K^j_{A_2}e^{\frac{j}{2b_2}[a_2n_2^2+2n_2(r_2-u_2)-2u_2(d_2r_2-b_2s_2)+d_2u_2^2]}dpdqdw\\\
\nonumber &&=\int_{{\mathbf R^2}}\left\{\left[\int_{{\mathbf R^2}}\right.\right.\\\
\nonumber &&\;  K^i_{A_1}e^{\frac{i}{2b_1}[a_1p_1^2+2p_1(r_1-u_1)-2u_1(d_1r_1-b_1s_1)+d_1u_1^2]}\\\
\nonumber && \; f\left(w+\frac{p}{2}\right)\overline{f\left(w-\frac{p}{2}\right)}\\\
\nonumber &&\; \times \left.K^j_{A_2}e^{\frac{j}{2b_2}[a_2p_2^2+2p_2(r_2-u_2)-2u_2(d_2r_2-b_2s_2)+d_2u_2^2]}dp\right]\\\
\nonumber&& \;\times \int_{{\mathbf R^2}}e^{\frac{i}{2b_1}[a_1q_1^2-2q_1(r_1-u_1)]}g\left(t-w+\frac{q}{2}\right)\\\
\nonumber &&\; \times \left. \overline{g\left(t-w-\frac{q}{2}\right)}e^{\frac{j}{2b_2}[a_2q_2^2-2q_2(r_2-u_2)]}dq\right\}\\\
\label{4.8}&&\;\times e^{-i\frac{a_1}{b_1}\left(4w_1(t_1-w_1)\right)}e^{-j\frac{a_2}{b_2}\left(4w_2(t_2-w_2)\right)}dw
\end{eqnarray}
Now multiply  (\ref{4.8})both sides by $K^i_{A_1}e^{\frac{i}{2b_1}[d_1u_1^2-2u_1(d_1r_1-b_1s_1)]}$ and $K^j_{A_2}e^{\frac{j}{2b_2}[d_2u_2^2-2u_2(d_2r_2-b_2s_2)]}$,we get
\begin{eqnarray}
\nonumber &&K^i_{A_1}e^{\frac{i}{2b_1}[d_1u_1^2-2u_1(d_1r_1-b_1s_1)]}\\\
\nonumber &&\; K^j_{A_2}e^{\frac{j}{2b_2}[d_2u_2^2-2u_2(d_2r_2-b_2s_2)]}\mathcal W^{A_1,A_2}_{f\star g}(t,u)\\\
\nonumber &&=\int_{{\mathbf R^2}}e^{-i\frac{a_1}{b_1}\left(4w_1(t_1-w_1)\right)}\mathcal W^{A_1,A_2}_{f,f}(w,u)\\\
\label{4.9}&&\mathcal W^{A_1,A_2}_{g,g}(t-w,u)e^{-j\frac{a_2}{b_2}\left(4w_2(t_2-w_2)\right)}dw
\end{eqnarray}
Now using (\ref{4.6}) and (\ref{4.7}) in (\ref{4.9}) we get,
\begin{eqnarray*}
&&\mathcal W^{A_1,A_2}_{f\star g}(t,u)\\\
&&=\sqrt{2\pi b_1 i}e^{\frac{-i}{2b_1}[d_1(u_1^2+r_1^2)-2u_1(d_1r_1-b_1s_1)]}\\\
&& \; \times\left\{\int_{\mathbf R^2}e^{-i\frac{a_1}{b_1}(4w_1(t_1-w_1))}\mathcal W^{A_1,A_2}_{f,f}(w,u)\mathcal W^{A_1,A_2}_{g,g}(t-w,u)\right.\\\
&& \;\left. e^{-j\frac{a_2}{b_2}(4w_2(t_2-w_2))}dw\right\}\\\
&& \;\times\sqrt{2\pi b_2 j}e^{\frac{-j}{2b_2}[d_2(u_2^2-r_2^2)-2u_2(d_2r_2-b_2s_2)]}
\end{eqnarray*}
which completes the proof of theorem. \end{proof} 
{\bf Consequences of theorem \ref{th 4.2}.}
\begin{enumerate}
\item  Changing parameter $A_i =\left[\begin{array}{cccc}a_i & b_i &| & r_i\\c_i & d_i &| & s_i \\\end{array}\right] $ ,$i=1,2$  to $A_i =\left[\begin{array}{cccc}a_i & b_i &| & 0\\c_i & d_i &| & 0 \\\end{array}\right] $ ,$i=1,2$, then the  Theorem \ref{th 3.5} reduces to convolution
theorem of the WVD-QLCT as follows:
\begin{eqnarray}
\nonumber && \mathcal W^{A_1,A_2}_{f\star g}(t,u)\\\
\nonumber && =\sqrt{2\pi b_1 i}e^{\frac{-i}{2b_1}d_1u_1^2}\sqrt{2\pi b_2 j}e^{\frac{-j}{2b_2}d_2u_2^2}\\\
\nonumber && \times\left\{
\int_{\mathbf R^2}e^{-i\frac{a_1}{b_1}(4w_1(t_1-w_1))}\mathcal W^{A_1,A_2}_{f,f}(w,u)\mathcal W^{A_1,A_2}_{g,g}(t-w,u)\right.\\\
\label{4.10}&& \left. e^{-j\frac{a_2}{b_2}(4w_2(t_2-w_2))}dw
\right\}
\end{eqnarray}
where $ \mathcal W^{A_1,A_2}_{f,f}$ and $\mathcal W^{A_1,A_2}_{g,g}$ is the WVD in the QLCT domain of a signal $f$ and
$g$, respectively.
\item  Changing parameter $A_i =\left[\begin{array}{cccc}a_i & b_i &| & r_i\\c_i & d_i &| & s_i \\\end{array}\right] $ ,$i=1,2$  to $A_i =\left[\begin{array}{cccc}0 & 1 &| & 0\\-1 & 0 &| & 0 \\\end{array}\right] $ ,$i=1,2$, then the  Theorem \ref{th 3.5} reduces to convolution
theorem of the WVD in Quaternion Domain as follows:
\begin{eqnarray}
\nonumber \mathcal W_{f\star g}(t,u)&=&\sqrt{2\pi  i}\left\{
\int_{\mathbf R^2}\mathcal W_{f,f}^{i,j}(w,u)\mathcal W_{g,g}^{i,j}(t-w,u)dw
\right\}\\\
\label{4.11}&&\sqrt{2\pi  j}
\end{eqnarray}
where $ \mathcal W_{f,f}^{i,j}$ and $\mathcal W_{g,g}^{i,j}$ is the WVD in the Quaternion domain of a signal $f$ and
$g$, respectively.
\end{enumerate}
Next, we will derive the correlation theorem in the WVD-QOLCT. Let us define the
correlation for the QOLCT.
\begin{definition}\label{def 4.3}
  For any two quaternion functions $f,g \in L^2(\mathbf R^2,\mathbf H),$ we define the correlation operator of the QOLCT
as
\begin{equation}\label{4.12}(f\circ g)(t)=\int_{\mathbf R^2}e^{i\frac{a_1}{b_1}2z_1(z_1+t_1)}\overline{f(z)}g(z+t)e^{j\frac{a_2}{b_2}2z_2(z_2+t_2)} dz\end{equation}
\end{definition}
Now, we reap a consequence of the above definition .
\begin{theorem}\label{th 4.4}{\bf (WVD-QOLCT Correlation).} For any two quaternion functions $f,g \in L^2{(\mathbf R^2,\mathbf H)},$ the following result holds
\begin{eqnarray}
\nonumber &&\mathcal W^{A_1,A_2}_{f\circ g}(t,u)\\\
\nonumber&& =\sqrt{2\pi b_1 i}e^{\frac{-i}{2b_1}[d_1(u_1^2+r_1^2)+2u_1(d_1r_1-b_1s_1)]}\\\
\nonumber &&\times\left\{\int_{\mathbf R^2}e^{i\frac{a_1}{b_1}(4w_1(t_1+w_1))}\mathcal W^{A_1,A_2}_{f,f}(w,-u)\right.\\\
\nonumber && \left. \mathcal W^{A_1,A_2}_{g,g}(t+w,u)e^{j\frac{a_2}{b_2}(4w_2(t_2+w_2))}dw\right\}\\\
\label{4.13}&&\times\sqrt{2\pi b_2 j}e^{\frac{-j}{2b_2}[d_2(u_2^2+r_2^2)+2u_2(d_2r_2-b_2s_2)]}
\end{eqnarray}
\end{theorem}
\begin{proof} Applying the definition of the WVD-QOLCT  we have
\begin{eqnarray}
\nonumber && \mathcal W^{A_1,A_2}_{f\circ g}(t,u)\\\
\nonumber &&=\int_{{\mathbf R}^2}{K^{i }_{A_1}(n_1,u_1)}\left[(f\circ g)(t+\frac{n}{2})\right]\\\
\label{4.14}&&\left[\overline{f}\circ\overline{g}(t-\frac{n}{2})\right]K^{j }_{A_2}(n_2,u_2)dn\end{eqnarray}
Now using definition \ref{def 4.3} in (\ref{4.14}) we have
\begin{eqnarray}
\nonumber &&\mathcal W^{A_1,A_2}_{f\circ g}(t,u)\\\
\nonumber &&=\int_{{\mathbf R}^2}{K^{i }_{A_1}(n_1,u_1)}\left\{\int_{{\mathbf R}^2}e^{i\frac{a_1}{b_1}2z_1(z_1+(t_1+\frac{n_1}{2}))}\overline{f(z)}\right. \\\
\nonumber && g(z+t+\frac{n}{2})e^{j\frac{a_2}{b_2}2z_2(z_2+(t_2+\frac{n_2}{2}))}dz\\\
\nonumber &&\;\times\int_{{\mathbf R}^2}e^{i\frac{a_1}{b_1}2\gamma_1(\gamma_1+(t_1-\frac{n_1}{2}))}\overline{\overline{f(\gamma)}g(\gamma+(t-\frac{n}{2}))}\\\
\label{4.15}&& \left. e^{j\frac{a_2}{b_2}2\gamma_2(\gamma_2+(t_2-\frac{n_2}{2}))}d\gamma \right\}{K^{j }_{A_2}(n_2,u_2)}dn
\end{eqnarray}
Now with the help of (\ref{4.6}) and (\ref{4.7}),we have from (\ref{4.15})
\begin{eqnarray*}
&&\mathcal W^{A_1,A_2}_{f\circ g}(t,u)\\\
&&=\int_{{\mathbf R^6}}K^i_{A_1}e^{\frac{i}{2b_1}[a_1n_1^2+2n_1(r_1-u_1)-2u_1(d_1r_1-b_1s_1)+d_1u_1^2]}\\\
&& \;e^{i\frac{a_1}{b_1}2z_1(z_1+(t_1+\frac{n_1}{2}))}\\\
&&\;\times \overline{f(z)}g(z+(t+\frac{n}{2}))
 e^{j\frac{a_2}{b_2}2z_2(z_2+(t_2+\frac{n_2}{2}))}
 e^{i\frac{a_1}{b_1}2\gamma_1(\gamma_1+(t_1-\frac{n_1}{2}))}\\\
&&\;\times \overline{\overline{f(\gamma)}g(\gamma+(t-\frac{n}{2}))}e^{j\frac{a_2}{b_2}2\gamma_2(\gamma_2+(t_2-\frac{n_2}{2}))}\\\
&&\;\times K^j_{A_2}e^{\frac{j}{2b_2}[a_2n_2^2+2n_2(r_2-u_2)-2u_2(d_2r_2-b_2s_2)+d_2u_2^2]}dzd\gamma dn
\end{eqnarray*}
Setting $z_i=w_i+\frac{p_i}{2},\gamma_i=w_i-\frac{p_i}{2},i=1,2$, we get
\begin{eqnarray}
\nonumber &&\mathcal W^{A_1,A_2}_{f\circ g}(t,u)\\\
\nonumber &&=\int_{{\mathbf R^6}}K^i_{A_1}e^{\frac{i}{2b_1}[a_1n_1^2+2n_1(r_1-u_1)-2u_1(d_1r_1-b_1s_1)+d_1u_1^2]}\\\
\nonumber &&e^{i\frac{a_1}{b_1}2\left(w_1+\frac{p_1}{2}\right)\left((t_1+\frac{n_1}{2})+(w_1+\frac{p_1}{2})\right)}\\\
\nonumber &&\; \times \overline{f\left(w+\frac{p}{2}\right)}g\left((t+\frac{n}{2})+(w+\frac{p}{2})\right)\\\
\nonumber &&e^{j\frac{a_2}{b_2}2\left(w_2+\frac{p_2}{2}\right)\left((t_2+\frac{n_2}{2})+(w_2+\frac{p_2}{2})\right)}\\\
\nonumber && \; \times e^{i\frac{a_1}{b_1}2\left(w_1-\frac{p_1}{2}\right)\left((t_1-\frac{n_1}{2})+(w_1-\frac{p_1}{2})\right)}f\left(w-\frac{p}{2}\right)\\\
\nonumber &&\overline{g\left((t-\frac{n}{2})+(w-\frac{p}{2})\right)}\\\
\nonumber &&\times e^{j\frac{a_2}{b_2}2\left(w_2-\frac{p_2}{2}\right)\left((t_2-\frac{n_2}{2})+(w_2-\frac{p_2}{2})\right)}K^j_{A_2}\\\
\label{4.16} &&e^{\frac{j}{2b_2}[a_2n_2^2+2n_2(r_2-u_2)-2u_2(d_2r_2-b_2s_2)+d_2u_2^2]}dpdq dw
\end{eqnarray}
Now put $n_i=q_i-p_i,i=1,2$ and on following the same procedure as followed in previous Theorem \ref{th 4.2}, we have from (\ref{4.16})
\begin{eqnarray}
\nonumber &&\mathcal W^{A_1,A_2}_{f\circ g}(t,u)\\\
\nonumber &&=\int_{{\mathbf R^2}}\left[\int_{{\mathbf R^2}}e^{\frac{i}{2b_1}[a_1p_1^2-2p_1(r_1-u_1)]}\overline{f\left(w+\frac{p}{2}\right)}f\left(w-\frac{p}{2}\right)\right.\\\
\nonumber && \left. e^{\frac{j}{2b_2}[a_2p_2^2-2p_2(r_2-u_2)]}dp\right]\\\
\nonumber && \times \left[\int_{{\mathbf R^2}}K^i_{A_1}e^{\frac{i}{2b_1}[a_1q_1^2+2q_1(r_1-u_1)-2u_1(d_1r_1-b_1s_1)+d_1u_1^2]}\right.\\\
\nonumber &&\times g\left(t+w+\frac{q}{2}\right) \overline{g\left(t+w-\frac{q}{2}\right)}\\\
\nonumber &&\times\left. K^j_{A_2}e^{\frac{j}{2b_2}[a_2q_2^2+2q_2(r_2-u_2)-2u_2(d_2r_2-b_2s_2)+d_2u_2^2]}dq\right]\\\
\label{4.17}&&\times e^{i\frac{a_1}{b_1}\left(4w_1(t_1+w_1)\right)}e^{j\frac{a_2}{b_2}\left(4w_2(t_2+w_2)\right)}dw
\end{eqnarray}
On multiplying (\ref{4.17}) both sides by\\ $K^i_{A_1}e^{\frac{i}{2b_1}[d_1u_1^2-2u_1(d_1r_1-b_1s_1)+4p_1(r_1-u_1)]}$ and \\ $K^j_{A_2}e^{\frac{j}{2b_2}[d_2u_2^2-2u_2(d_2r_2-b_2s_2)+4p_2(r_2-u_2)]}$,we get
\begin{eqnarray}
\nonumber && K^i_{A_1}e^{\frac{i}{2b_1}[d_1u_1^2-2u_1(d_1r_1-b_1s_1)+4p_1(r_1-u_1)]}\\\
\nonumber && K^j_{A_2}e^{\frac{j}{2b_2}[d_2u_2^2-2u_2(d_2r_2-b_2s_2)+4p_2(r_2-u_2)]}\mathcal W^{A_1,A_2}_{f\circ g}(t,u)\\\
\nonumber &&=\int_{{\mathbf R^2}}e^{-i\frac{a_1}{b_1}\left(4w_1(t_1-w_1)\right)}\mathcal W^{A_1,A_2}_{f,f}(w,u)\\\
\label{4.18}&&\mathcal W^{A_1,A_2}_{g,g}(t-w,u)e^{-j\frac{a_2}{b_2}\left(4w_2(t_2-w_2)\right)}dw
\end{eqnarray}
Now using (\ref{4.6}) and (\ref{4.7}) in (\ref{4.18}) we obtain,
\begin{eqnarray*}
&&\mathcal W^{A_1,A_2}_{f\star g}(t,u)\\\
&&=\sqrt{2\pi b_1 i}e^{\frac{-i}{2b_1}[d_1(u_1^2+r_1^2)+2u_1(d_1r_1-b_1s_1)]}\\\
&&\times\left\{\int_{\mathbf R^2}e^{i\frac{a_1}{b_1}(4w_1(t_1+w_1))}\mathcal W^{A_1,A_2}_{f,f}(w,-u)\mathcal W^{A_1,A_2}_{g,g}(t+w,u)\right.\\\
&& \times\left. e^{j\frac{a_2}{b_2}(4w_2(t_2+w_2))}dw\right\}\sqrt{2\pi b_2 j}e^{\frac{-j}{2b_2}[d_2(u_2^2+r_2^2)+2u_2(d_2r_2-b_2s_2)]}
\end{eqnarray*}
which completes the proof of theorem. \end{proof}
{\bf Consequences of Theorem \ref{th 4.4}.}
\begin{enumerate}
\item  Changing parameter $A_i =\left[\begin{array}{cccc}a_i & b_i &| & r_i\\c_i & d_i &| & s_i \\\end{array}\right] $ ,$i=1,2$  to $A_i =\left[\begin{array}{cccc}a_i & b_i &| & 0\\c_i & d_i &| & 0 \\\end{array}\right] $ ,$i=1,2$, then the  Theorem \ref{th 4.4} reduces to  correlation 
theorem of the WVD-QLCT as follows:
\begin{eqnarray*}
&&\mathcal W^{A_1,A_2}_{f\circ g}(t,u)\\\
&&=\sqrt{2\pi b_1 i}e^{\frac{-i}{2b_1}d_1u_1^2}\sqrt{2\pi b_2 j}e^{\frac{-j}{2b_2}d_2u_2^2}\\\
&&\times\left\{
\int_{\mathbf R^2}e^{i\frac{a_1}{b_1}(4w_1(t_1+w_1))}\mathcal W^{A_1,A_2}_{f,f}(w,-u)\right.\\\
&& \left.\mathcal W^{A_1,A_2}_{g,g}(t+w,u)e^{j\frac{a_2}{b_2}(4w_2(t_2+w_2))}dw
\right\}
\end{eqnarray*}
where $ \mathcal W^{A_1,A_2}_{f,f}$ and $\mathcal W^{A_1,A_2}_{g,g}$ is the WVD in the QLCT domain of a signal $f$ and
$g$, respectively.
\item Changing parameter $A_i =\left[\begin{array}{cccc}a_i & b_i &| & r_i\\c_i & d_i &| & s_i \\\end{array}\right] $ ,$i=1,2$  to $A_i =\left[\begin{array}{cccc}0 & 1 &| & 0\\-1 & 0 &| & 0 \\\end{array}\right] $ ,$i=1,2$, then the  Theorem 4.4 reduces to correlation
theorem of the WVD in Quaternion Domain as follows:
\begin{eqnarray*}
\mathcal W_{f\circ g}(t,u)&=&\sqrt{2\pi  i}\left\{
\int_{\mathbf R^2}\mathcal W_{f,f}^{i,j}(w,-u)\mathcal W_{g,g}^{i,j}(t+w,u)dw
\right\}\sqrt{2\pi  j}
\end{eqnarray*}
where $ \mathcal W_{f,f}^{i,j}$ and $\mathcal W_{g,g}^{i,j}$ is the WVD in the Quaternion domain of a signal $f$ and
$g$, respectively.
\end{enumerate}
\begin{acknowledgements}
This work is supported by the UGC-BSR Research Start Up Grant(No. F.30-498/2019(BSR)) provided by UGC, Govt. of India.
\end{acknowledgements}
 \section*{Conflict of interest}
 The authors declare that they have no conflict of interest.
\newpage


\begin{thebibliography}{}
\bibitem{1} M. Bahri, R. Ashino, and R. Vaillancourt, \emph{Convolution theorems for quaternion Fourier transform: properties and applications},  Abst.  Appl Anal.  2013, Article ID 162769.
\bibitem{2} M. Bahri, E. S. M. Hitzer, A. Hayashi, and R. Ashino,  \emph{An uncertainty principle for quaternion Fourier transform,}  Comps.  Maths with Appl., 56( 9)  2398–2410(2008).
\bibitem{3} M. Bahri, \emph{Correlation theorem for Wigner-Ville distribution}, Far East J.  Math.  Sci. 80(1)  123– 133(2013) .
\bibitem{4} R. F. Bai, B. Z. Li, and Q. Y. Cheng, \emph{Wigner-Ville distribution associated with the linear canonical transform},  J.  Appl.  Maths,  2012, Article ID 740161.
\bibitem{5} L. Debnath, B. V. Shankara, and N. Rao, \emph{On new two- dimensional Wigner-Ville nonlinear integral transforms and their basic properties}, Int. Trans. Sp.  Funct. , 21(3) 165–174(2010).
\bibitem{6} W. B. Gao and B.  Z. Li, \emph{Convolution and correlation theorems for the windowed offset linear canonical transform} arxiv: 1905.01835v2 [math.GM](2019)
\bibitem{7} X.  Guanlei,  W. Xiaotong, X.  Xiaogang,  \emph{Uncertainty inequalities for linear canonical transform}. IET Signal Process. 3(5) 392–402 (2009)
\bibitem{8} Y. El Haoui S. and S. Fahlaoui , Generalized Uncertainty Principles associated with the Quaternionic Offset Linear Canonical Transform, https://arxiv.org/abs/1807.04068v1.
\bibitem{9} E. M. S. Hitzer, \emph{Quaternion Fourier transform on quaternion fields and generalizations,}  Adv.  Appl.  Clifford Algs ,  17(3) 497–517(2007) .
\bibitem{10} H. Y. Huo, W. C. Sun, L. Xiao, \emph{Uncertainty principles associated with the offset linear canonical transform} Mathl. Methods  Appl.  Scis.   42(2)  466-474(2019) .
\bibitem{11} K. I. Kou, Jian-YuOu, J. Morais, On uncertainty principle for quaternionic linear canonical transform, Abstr. Appl. Anal.,2013 (Article ID 725952) (2013) 14pp
\bibitem{12} K. I. Kou, J. Morais and Y. Zhang, Generalized prolate spheroidal wave functions for offset linear canonical transform in clifford analysis, Mathematical Methods in the Applied Sciences 36 (9) 1028-1041(2013).
\bibitem{13} Y. G. Li, B. Z Li and H. F. Sun, \emph{Uncertainty principle for Wigner-Ville distribution associated with the linear canonical transform},  Abstr. Appl. Anal., 2014, Article ID 470459.
\bibitem{14} Y.E. Song, X.Y. Zhang, C.H. Shang, H.X. Bu, X.Y. Wang, \emph{The Wigner-Ville distribution based on the linear canonical transform and its
applications for QFM signal parameters estimation}, J.  App. Maths  (2014) 8 pages.
\bibitem{15} D. Urynbassarova, B. Z. Li, and R. Tao, \emph{The Wigner-Ville distribution in the linear canonical transform domain,} IAENG Int. J.  Appl.  Maths. 46 (4)  559-563(2016).
\bibitem{16} D. Urynbassarova, B.Zhao, R.Tao, \emph{ Convolution and Correlation Theorems for Wigner-Ville Distribution Associated with the Offset Linear Canonical Transform}. Int.  J.  Light  Elect. Optics http://dx.doi.org/10.1016/j.ijleo.2017.08.099
\bibitem{17} D. Wei, Q. Ran, and Y. Li,  \emph{A convolution and correlation theorem for the linear canonical transform and its application}, Circuits Syst. Signal Process., 31(1)  301–312(2012).
\bibitem{18} D. Wei, Q. Ran, and Y. Li, \emph{New convolution theorem for the linear canonical transform and its translation invariance property}, Optik., 123(16) 1478–1481(2012).
\bibitem{19} Z. C. Zhang, \emph{Sampling theorem for the short-time linear canonical transform and its applications}. Signal Proces. 113138-146( 2015) .
%
%

\end{thebibliography}


\end{document}